\makeatletter
\providecommand\UseOneTimeHook[2][]{}
\makeatother

\documentclass[a4paper,twocolumn,11pt,submission]{quantumarticle}
\makeatletter
\providecommand{\@afterenddocumenthook}{}
\makeatother
\pdfoutput=1
\usepackage{amsmath,amssymb,amsthm,mathtools}
\usepackage{physics}
\usepackage{microtype}
\usepackage{enumitem}
\usepackage{authblk}
\usepackage{tikz}
\usepackage{quantikz}
\usetikzlibrary{arrows.meta,positioning,fit,calc}
\usepackage{hyperref}
\hypersetup{colorlinks=true,linkcolor=blue,citecolor=blue,urlcolor=blue}

\DeclareUnicodeCharacter{03B5}{$\epsilon$}
\DeclareUnicodeCharacter{03F5}{$\epsilon$}
\DeclareUnicodeCharacter{000F}{$\epsilon$}
\renewcommand{\varepsilon}{\epsilon}

\title{Counterfactual Local Friendliness: An $\epsilon$-Bounded Interaction-Free Paradox and a Disturbance-Robust Three-Box Inequality}
\author[1]{Maximilian Ralph Peter von Liechtenstein}
\affil[1]{Independent Researcher}
\date{30 August 2025}

\newtheorem{theorem}{Theorem}
\newtheorem{lemma}{Lemma}
\newtheorem{definition}{Definition}

\begin{document}
\maketitle

\begin{abstract}
We introduce a new paradox, which we call \emph{Counterfactual Local Friendliness (CLF)}: a Wigner’s-friend–type logical collision in which every decisive inference is obtained by interaction-free flags whose disturbance on the probed object is bounded by a tunable parameter $\epsilon$. Under (Q) universal unitarity for outside observers, (S) single-outcome facts, (C) cross-agent consistency, and (IF-$\epsilon$) $\epsilon$-counterfactuality of the friends’ internal modules, quantum theory predicts a nonzero post-selected event that forces mutually incompatible certainties about a single upstream variable --- without appealing to absorptive or projective in-lab measurements.

We also derive an $\epsilon$-IF three-box \emph{noncontextual bound}: any single-world, noncontextual model satisfying exclusivity and $\epsilon$-stability must obey $P_A + P_B \le 1 + K\,\epsilon$, while quantum theory yields $P_A = P_B = 1$, violating the bound for arbitrarily small $\epsilon$. Together these results isolate what is paradoxical about counterfactual phenomena: not energy exchange with the probed system, but the incompatibility of agent-level facts in single-world narratives.
\end{abstract}


\section{Introduction}
Interaction-free measurement (IFM) shows that a system can reveal the presence of a “live” absorber without absorbing the probe. Classic examples include the Elitzur–Vaidman bomb tester \cite{Elitzur1993} and its Zeno-boosted variants \cite{Kwiat1995,Kwiat1999}, as well as Hardy’s paradox \cite{Hardy1992}. Separately, Wigner’s-friend and “local friendliness” arguments demonstrate tensions between universal unitarity and single-outcome, agent-independent facts. Here we combine these strands and add a control parameter $\epsilon$ that quantifies how close to “no interaction at all” an IFM can be, operationally.

We pursue two goals: (i) exhibit a logical collision of certainties generated by IFM-based friends and outside Wigners, and (ii) give a noncontextual inequality for a three-box IFM that is violated by quantum predictions even as $\epsilon \to 0$. The novelty is the $\epsilon$-counterfactual formalization and its use in closing the intuitive “maybe it interacted a little” escape route.

\subsection{Notation and conventions}
$\epsilon$ denotes the disturbance bound. The trace distance between states $\rho,\sigma$ is $T(\rho,\sigma) = \frac{1}{2}\norm{\rho - \sigma}_1$. A \emph{decisive outcome} is the outcome whose occurrence triggers an inference (e.g., a Dark flag) and to which the $\epsilon$-bound on disturbance applies.

A \emph{flag qubit} takes values $D$ (Dark) or $B$ (Bright) written by a unitary IFM oracle; no photonic clicks are required in our theory model.

We use the following assumption labels for clarity: (Q) universal unitarity for outside observers; (S) single-outcome facts; (C) cross-agent consistency; (IF-$\epsilon$) $\epsilon$-counterfactuality (Def.~\ref{def:ecf}) of the friends’ internal modules; (IF-$\epsilon$-stab) $\epsilon$-stability (Def.~\ref{def:stab}) used in Sec.~\ref{sec:threebox}.

“Probability nonzero” means strictly greater than zero under the ideal unitary model; robustness to small noise is handled in Appendix B.

\subsection{Statement of novelty and contributions}
\textbf{New paradox (primary).} We introduce \emph{Counterfactual Local Friendliness (CLF)} --- to our knowledge the first Wigner’s-friend no-go in which all decisive inferences are made via interaction-free flags with a provable $\epsilon$-bound on disturbance to the probed system. Prior Wigner’s-friend / local-friendliness arguments rely on absorptive/projective measurements inside the labs; our $\epsilon$-IF formulation removes “measurement disturbance of the object” as an explanatory loophole.

\textbf{New quantitative bound.} We formulate an $\epsilon$-stable three-box inequality ($P_A + P_B \le 1 + K\,\epsilon$) that is violated by quantum mechanics for arbitrarily small $\epsilon$, giving a disturbance-robust, device-agnostic witness of nonclassicality.

\textbf{New formulations across paradox families.} We provide $\epsilon$-IF versions of GHZ all-versus-nothing, Peres–Mermin state-independent contextuality, Leggett–Garg macrorealism, and a sheaf-theoretic “no global history” statement --- each with explicit $\epsilon$ (and $\sqrt{\delta}$) slacks --- showing broad portability of the $\epsilon$-counterfactual framework.

\textbf{Engineering bridge.} We specify how $\epsilon$ can be estimated from observables (visibility, loss, leakage) and composed over rounds, enabling $\epsilon$-certified imaging, $\epsilon$-aware QEC scheduling, and audit logs.

\textbf{Claim of novelty.} To the best of our knowledge, CLF as defined here constitutes a new paradox: a no-go for single-world, agent-independent facts where every decisive inference is counterfactual and $\epsilon$-bounded. We also believe the $\epsilon$-stable three-box inequality in this exact form is new.

\section[$\epsilon$-counterfactuality: formal model of IFM]{\boldmath$\epsilon$-counterfactuality: formal model of interaction-free measurement}
\label{sec:ecfmodel}
\subsection{Preliminaries}
Let $S$ be the system (photon/path), $B$ the “bomb,” and $D$ a detector/pointer. An IFM instrument with two outcomes ($x$, $\bar{x}$) is a CPTP map $\{ \mathcal{E}_x,\, \mathcal{E}_{\bar{x}}\}$ on $(S\otimes B)$, with classical register $X$ output to the outside (indicating which outcome occurred).

\subsection{Definition of $\epsilon$-counterfactual IFM}
\label{def:ecf}
{\setcounter{definition}{1}%
\begin{definition}[$\epsilon$-counterfactual IFM]\label{def:ecf2}
An IFM $\{\mathcal{E}_x\}$ is \emph{$\epsilon$-counterfactual} for outcome $x$ on bomb states in a set $\mathcal{B}$ if for all $\rho_B \in \mathcal{B}$ and all system states $\rho_S$, the reduced bomb state change obeys 
\begin{equation}
\label{eq:ecf-def}
\left\Vert \Tr_S\!\big[\mathcal{E}_x(\rho_S \otimes \rho_B)\big] - \rho_B \right\Vert_1 \;\le\; \epsilon~.
\end{equation}
When $\epsilon=0$ the bomb is left exactly unchanged whenever outcome $x$ is registered.
\end{definition}}

\subsection{Oracle description and unitary gadget}
A practical implementation of an $\epsilon$-counterfactual measurement uses an \emph{IFM oracle}: a unitary gadget acting on a mediator qubit and the bomb such that a “Dark” outcome indicates the bomb was in state $|1\rangle$ (“live”) while a “Bright” outcome indicates the bomb was $|0\rangle$ (“dud”). One realization is given in Appendix A. In essence, the mediator qubit passes through an interference cycle that includes a phase flip controlled by the bomb; at the output, it coherently flags $D$ if the bomb was live, without any photon absorption. This is a \emph{unitary flag} --- no absorber clicks are needed. Figure~\ref{fig:oracle} shows a circuit cartoon of such an IFM oracle.

\begin{figure}[!htbp]
\centering
\resizebox{0.95\linewidth}{!}{
\begin{quantikz}[row sep={0.8cm,between origins}, column sep=0.7cm]
\lstick{$b$ (bomb)}      & \qw       & \ctrl{1} & \qw       & \ctrl{2} & \qw \\
\lstick{$S$ (mediator)}  & \gate{H}  & \gate{Z} & \gate{H}  & \qw       & \qw \\
\lstick{$W$ (flag)}      & \qw       & \qw      & \qw       & \targ{}   & \meter{} 
\end{quantikz}
}
\caption{IFM oracle with unitary flag.
The mediator $S$ undergoes $H\!-\!Z^b\!-\!H$: a controlled-$Z$ from the bomb $b$ acts between the two Hadamards, leaving $b$ untouched.
A final CNOT ($b\!\rightarrow\!W$) toggles the flag to Dark iff $b{=}1$ (the $\epsilon$-counterfactual outcome).}
\label{fig:oracle}
\end{figure}
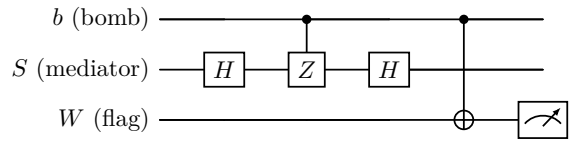

\subsection{Channel-level $\epsilon$ (diamond norm) and metric choices}
Our baseline $\epsilon$ bounds the bomb’s state change on decisive outcomes. For completeness we also define a channel-level notion. Let $\mathcal{E}_x$ be the CPTP map induced on the bomb conditional on decisive outcome $x$. Define the channel $\epsilon$ as a diamond-norm deviation from identity: $\norm{\mathcal{E}_x - \mathrm{id}}_{\diamond} \le \epsilon_{\diamond}$. This implies a state-level bound for all inputs and ancillas by definition, with constants linking $\epsilon_{\diamond}$ and the trace-distance bound $\epsilon$ used elsewhere. For small perturbations, the Bures angle and fidelity provide equivalent bounds up to second order. In applications, one may report either (i) a direct state-level $\epsilon$ (empirical and simple) or (ii) a worst-case channel $\epsilon_{\diamond}$ (more conservative). We keep results at the state level for simplicity.

\section{CLF paradox: two friends with IFM and one coin}
The CLF scenario involves two spatially separated labs (A and B), each containing a friend who uses an IFM oracle to detect a local “bomb” qubit without disturbing it (beyond $\epsilon$). An initial coin qubit $C$ is prepared and coherently distributed to the labs such that their inferences may conflict. We outline the logical structure below; a concrete unitary realization is given in Appendix A.

Each lab’s IFM gadget writes a flag qubit $W_j$ ($j=A,B$) which can be $D$ (dark) or $B$ (bright). A Dark flag $W_j=D$ implies the local bomb $b_j$ was in state $1$ (live) with certainty (by the IFM oracle design); Bright implies $b_j=0$ (dud). In our setup, $b_A = C_A$ is simply the $Z$-basis value of the coin (hence $b_A=1$ corresponds to coin $C=0$), while $b_B$ is the $X$-basis value of another coin register $C_B$ entangled with $C$ (hence $b_B=1$ corresponds to coin $C=1$). The two Wigners (outside observers) measure $W_A$ and $W_B$. We post-select on the event $W_A = D$ and $W_B = D$, which quantumly occurs with some probability $p>0$ (Appendix A).

In those runs, the outside agents infer $b_A = b_B = 1$. Chaining back through the coherent copy maps enforced inside the labs, cross-consistency (C) forces mutually incompatible assignments for the single upstream coin $C$: one branch of the reasoning makes it $C=0$ with certainty, the other $C=1$ with certainty --- a contradiction with single-outcome facts (S). This is the logical collision of certainties promised.

\subsection{Theorem 1: CLF no-go result}
\begin{theorem}[CLF No-Go]
Under assumptions (Q), (S), (C), and (IF-$\epsilon$), there exists a unitary-only protocol (as described above and in Appendix A) and some $\epsilon_0 > 0$ such that for all $0 \le \epsilon < \epsilon_0$, the post-selected event $\{W_A = D,\;W_B = D\}$ has nonzero probability and forces an inconsistent set of agent-level certainties about a single upstream coin $C$ (i.e., a logical collision).
\end{theorem}
\begin{proof}[Proof sketch]
The IFM oracle realizes a basis-dependent, disturbance-free inference in each lab: Dark implies bomb = live; Bright implies bomb = dud. The coherent coin feed-forward (Appendix A) reproduces the Frauchiger–Renner dependency graph with these IFM-style observables. Standard modal reasoning then yields contradictory certainties “$C=0$” and “$C=1$” on the same post-selected subset. Robustness to small $\epsilon$ follows by continuity (Appendix B).
\renewcommand\qedsymbol{} 
\end{proof}

\noindent \emph{Caveat.} Earlier drafts spoke of two clicks at dark ports. Here we make the outcome registers explicit as qubits $W_A$ and $W_B$; in the ideal oracle model both can be Dark simultaneously without violating particle number conservation.

Figure~\ref{fig:clf} illustrates the dependency graph of inferences in the CLF scenario. In the post-selected subset where both flags are Dark, the loop of certainties becomes inconsistent: $W_A = D$ implies $b_A=1$ which (by lab encoding) implies the coin $C=0$, whereas $W_B = D$ implies $b_B=1$ which implies $C=1$. This contradiction is the crux of CLF.
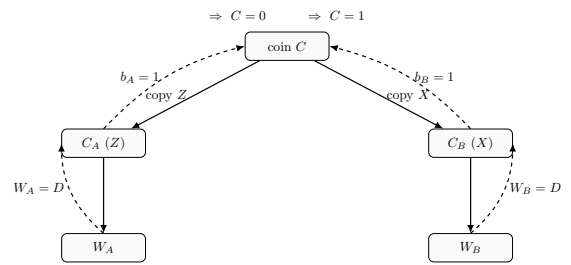
\begin{figure}[!htbp]
\centering
\resizebox{0.9\linewidth}{!}{
\begin{tikzpicture}[
  font=\small,
  node distance = 18mm and 26mm,
  var/.style   = {draw, rounded corners, fill=black!2, inner sep=2.5pt,
                  minimum width=22mm, minimum height=7.5mm},
  flow/.style  = {-{Latex[length=2.2mm,width=1.8mm]}, line width=0.9pt},
  infer/.style = {-{Latex[length=2.2mm,width=1.8mm]}, dashed, line width=0.9pt}
]
\node[var] (C)  {coin $C$};
\node[var, below left = of C] (CA) {$C_A$ ($Z$)};
\node[var, below right= of C] (CB) {$C_B$ ($X$)};
\node[var, below = 20mm of CA] (WA) {$W_A$};
\node[var, below = 20mm of CB] (WB) {$W_B$};

\draw[flow] (C)  -- node[left, pos=.52]  {copy $Z$} (CA);
\draw[flow] (C)  -- node[right,pos=.52]  {copy $X$} (CB);
\draw[flow] (CA) -- (WA);
\draw[flow] (CB) -- (WB);

\draw[infer] (WA.north) to[bend left=24] node[left,pos=.55] {$W_A=D$} (CA.west);
\draw[infer] (CA.north) to[bend left=16] node[left,pos=.48] {$b_A=1$} (C.west);
\node at ($(C)+(-13mm,8.2mm)$) {$\Rightarrow~C=0$};

\draw[infer] (WB.north) to[bend right=24] node[right,pos=.55] {$W_B=D$} (CB.east);
\draw[infer] (CB.north) to[bend right=16] node[right,pos=.48] {$b_B=1$} (C.east);
\node at ($(C)+( 13mm,8.2mm)$) {$\Rightarrow~C=1$};

\end{tikzpicture}%
}
\caption{CLF dependency graph. Nodes represent the coin $C$, lab copies $C_A,C_B$, and flags $W_A,W_B$.
Solid arrows show causal dependence (copy of $C$ into lab registers; flag generation).
Dashed arrows show logical inferences on the post-selected subensemble where both flags are $D$.
Each chain implies a different value for $C$ ($C=0$ from $W_A=D$, $C=1$ from $W_B=D$), yielding a contradiction.}
\label{fig:clf}
\end{figure}

\section{An \texorpdfstring{$\epsilon$}{$\epsilon$}-IF three-box inequality (quantitative paradox)}
\label{sec:threebox}
\subsection{Pre/post-selection}
Let $\{\ket{A},\ket{B},\ket{C}\}$ be orthonormal “box” states. 
Preselect and postselect
\begin{subequations}\label{eq:prepost}
\begin{align}
\ket{\psi_i} &= \tfrac{1}{\sqrt{3}} \bigl(\ket{A}+\ket{B}+\ket{C}\bigr), \\
\ket{\psi_f} &= \tfrac{1}{\sqrt{3}} \bigl(\ket{A}+\ket{B}-\ket{C}\bigr).
\end{align}
\end{subequations}

For the two–outcome test $\{\Pi_A,\, I-\Pi_A\}$, the Aharonov–Bergmann–Lebowitz (ABL) rule gives
\begin{equation}
P(\Pi_A = 1 \mid i,f) =
\dfrac{|\langle \psi_f | \Pi_A | \psi_i \rangle|^2}{%
  \substack{
    |\langle \psi_f | \Pi_A | \psi_i \rangle|^2 \\
    + |\langle \psi_f | (I-\Pi_A) | \psi_i \rangle|^2
  }}.
\label{eq:ABL}
\end{equation}

\subsection{$\epsilon$-counterfactual probes of A and B}
We implement tests of $\Pi_A$ and $\Pi_B$ with $\epsilon$-counterfactual IFM oracles: use a mediator qubit that interacts (via a unitary gadget as in Sec.~2) with an external “bomb” $b_A$ or $b_B$ placed in path $A$ or $B$ respectively. We tune the gadget (e.g. via a small beam-splitter reflectivity or Zeno cycles) such that the Dark outcome corresponds to $\Pi=1$ (i.e., the presence of the box is inferred). The $\epsilon$-counterfactuality ensures the bomb’s reduced state changes by at most $\epsilon$ on the decisive outcome.

The experiment is: prepare $|\psi_i\rangle$, choose context $A$ or $B$ to probe with the IFM (placing a bomb in arm $A$ or $B$ accordingly), record the flag outcome if decisive, and finally postselect on $|\psi_f\rangle$. We define operational probabilities $P_A$ and $P_B$ as the observed probabilities of the decisive (Dark) outcome in contexts $A$ and $B$ respectively, conditioned on successful postselection.

\subsection{Noncontextual, single-world bound}
\label{sec:noncontextual-bound}
\begin{definition}[$\epsilon$-stability]\label{def:stab}
Let $\mu(\cdot|A)$ and $\mu(\cdot|B)$ be the ontic distributions of the post-selected preparation when probing $A$ or $B$. The probing is \emph{$\epsilon$-stable} if the total-variation distance between these distributions is bounded by $K'\epsilon$: 
\begin{equation}
\mathrm{TV}\big(\mu(\cdot|A),\,\mu(\cdot|B)\big) \;\le\; K'\,\epsilon~,
\end{equation}
for some constant $K'$ independent of $\epsilon$. 
\end{definition}
Operationally, $\epsilon$-stability follows from $\epsilon$-counterfactuality plus data-processing (contractivity of trace distance) --- see Appendix F for a justification.

Now assume a preparation- and measurement-\emph{noncontextual}, single-world ontological model of this three-box experiment. Let the ontic space be $\Lambda$, and let $v(X) \in \{0,1\}$ denote the predetermined value of proposition $X$ in a given ontic state (e.g., $v(A)=1$ means the box $A$ is occupied in that ontic state). Single-world exclusivity means $v(A)+v(B)+v(C) = 1$ for all ontic states, and $v(A),v(B),v(C)\in\{0,1\}$. Partition $\Lambda$ into disjoint regions $\Lambda_A, \Lambda_B, \Lambda_C$ where $\Lambda_A$ is the set of ontic states with $v(A)=1$, etc. Then define indicator functions $\chi_A = 1_{\Lambda_A}$ and $\chi_B = 1_{\Lambda_B}$. The operational success probabilities can be written as integrals over these regions: 
\begin{align}
P_A &= \int \chi_A\, d\mu(\cdot|A), \label{eq:PA}\\
P_B &= \int \chi_B\, d\mu(\cdot|B). \label{eq:PB}
\end{align}
Exclusivity implies $\Lambda_A \cap \Lambda_B = \emptyset$ and $\Lambda_A \cup \Lambda_B \cup \Lambda_C = \Lambda$ (up to null sets). If switching between probing contexts changes the distribution by at most $O(\epsilon)$ (by $\epsilon$-stability), then:
\begin{equation}
\begin{split}
P_A+P_B \;\le\;& \int (\chi_A+\chi_B)\,d\mu(\cdot|A) \\
  &+ \tfrac{1}{2}\|\mu(\cdot|A)-\mu(\cdot|B)\|_1 \\
  \le\;& 1+K\epsilon~.
\end{split}
\label{eq:threebox}
\end{equation}
for some constant $K$ (e.g. $K=2K'$ if TV is defined without the $1/2$). This is the $\epsilon$-IF three-box inequality.

\subsection{Quantum violation}
Quantum mechanically, the ABL analysis yields $P_A = P_B = 1$ (Sec.~4.1). Hence 
\begin{equation}
P_A + P_B \;=\; 2 \;>\; 1 + K\,\epsilon~,
\end{equation}
violating the bound \eqref{eq:threebox} for arbitrarily small disturbance $\epsilon$. This provides a disturbance-robust, quantitative form of the three-box paradox that pins the failure on noncontextual single-world narratives rather than on measurement back-action.

\section{Discussion and falsifiable criteria}
\textbf{What is new.} (i) An operational $\epsilon$-counterfactual definition tailored to IFM as a disturbance bound on the probed object; (ii) a Wigner-friend contradiction where every agent-level certainty rests on such $\epsilon$-bounded, interaction-free flags; (iii) a compact three-box noncontextual bound whose violation persists as $\epsilon \to 0$, making disturbance-robustness explicit.

\textbf{Falsifiable target.} Any future operational theory that (a) reproduces dark-port certainties with $\epsilon$-bounded bomb disturbance and (b) enforces single-world, noncontextual facts must violate either (Q) or (C). Our results delimit this trade-off quantitatively via inequality \eqref{eq:threebox}.

\textbf{Relation to prior work.} Our CLF construction mirrors the structure of Wigner’s-friend / local-friendliness paradoxes, but replaces projective “inside” measurements by $\epsilon$-counterfactual IFM modules. Our inequality translates the three-box “double certainty” into a disturbance-robust noncontextual bound.

\section{Outlook (theoretical only)}
The $\epsilon$-counterfactual framework invites extensions: (i) counterfactual entanglement-swapping with $\epsilon$-bounded mediators; (ii) multi-box exclusivity graphs yielding state-independent $\epsilon$-IF contextuality inequalities; (iii) categorical semantics — model $\epsilon$-counterfactuality as a natural transformation preserving bomb objects.

\appendix

\section{Explicit unitaries for the CLF protocol}
We detail one explicit realization (at the qubit level) that is sufficient for Theorem 1. Let the upstream coin be $C$ initialized to $|+\rangle = (|0\rangle + |1\rangle)/\sqrt{2}$. Define lab registers $C_A, C_B$ (one per lab) and a routing qubit $R$. Use the unitary 

\begin{multline}
U: \quad |0\rangle_C \;\mapsto\; |0\rangle_{C_A}\,|+\rangle_{C_B}, \\
|1\rangle_C \;\mapsto\; |1\rangle_{C_A}\,|-\rangle_{C_B}~,
\end{multline}

which is an isometry from the single coin qubit into the two lab registers $C_A C_B$. We set the lab bombs such that $b_A = C_A$ (measured in the computational basis) and $b_B$ is the $X$-basis value of $C_B$. We also entangle a routing qubit $R$ with the coin so that, conditioned on $R$, a single probe qubit $S$ visits lab $L_A$ or $L_B$. This coherence makes joint dark-dark events possible (with probability $p>0$). Each lab applies the ideal IFM gadget (unitary oracle from Sec.~2) to its bomb. The outside Wigners measure the flag outputs $W_A, W_B$. A straightforward calculation shows $p>0$ and reproduces the inference loop described in Sec.~3. The above is just one convenient concrete construction (any Frauchiger–Renner style dependency graph suffices). 

\section{Robustness for small $\epsilon$}
If the decisive outcomes are $\epsilon$-counterfactual, the bomb’s state change is bounded in trace distance by $\epsilon$. By the Fuchs–van de Graaf inequality (relating trace distance to fidelity), this implies that all subsequent probability predictions change by at most $O(\sqrt{\epsilon})$. In other words, the certainties in the paradox become $1 - O(\sqrt{\epsilon})$ rather than exactly 1, and the logical implications carry through with small error terms that do not close the contradiction on the post-selected subset (for sufficiently small $\epsilon$).

\section{From gentle measurement to $\epsilon$-stability (sketch)}
\begin{lemma}[Gentleness $\implies \epsilon$-stability]
Consider two experimental contexts ($A$ and $B$) that differ only by the insertion or removal of an $\epsilon$-counterfactual IFM module acting on the target. Suppose the decisive (Dark) outcome in either context occurs with probability at least $1-\delta$. Then there exist constants $K_1, K_2$ such that 

\begin{equation}
\mathrm{TV}\big(\mu(\cdot|A),\,\mu(\cdot|B)\big) \;\le\; K_1\,\epsilon \;+\; K_2\,\sqrt{\delta}~.
\end{equation}
\end{lemma}
\begin{proof}[Proof idea]
The gentle measurement lemma states that if a two-outcome measurement is accepted with probability $\ge 1-\delta$, then the post-measurement state (conditional on acceptance) is $2\sqrt{\delta}$-close (in trace distance) to the pre-measurement state. Swapping context $B$ for context $A$ can be treated as such a “gentle” intervention together with the $\epsilon$-disturbance channel on the bomb, so by sequential application of gentle measurement and triangle inequality one obtains the stated bound (see also \cite{Aaronson2017} for a related bound).
\end{proof}
This result implies that the $\epsilon$-stability premise (Definition~\ref{def:stab}) can be justified in practice whenever $\delta$ can be made to scale with $\epsilon$. For example, one can tune a Zeno-style or interaction-free protocol to make the decisive outcome occur almost always (making $\delta$ as small as $O(\epsilon)$), in which case

\begin{equation}
\begin{split}
\mathrm{TV}\!\big(\mu(\cdot|A),\,\mu(\cdot|B)\big) 
  &\le K_1\,\epsilon \\
  &\quad + K_2\,\sqrt{\delta}~.
\end{split}
\label{eq:gentleness-stability}
\end{equation}

up to higher-order terms.

\section{Possibilistic CLF (modal formulation)}
In modal logic notation, let $\Diamond$ denote “possible” and $\square$ denote “necessary.” Consider the post-selected subset $S$ in which both $W_A$ and $W_B$ have the possible value Dark: $\Diamond(W_A=\text{Dark})$ and $\Diamond(W_B=\text{Dark})$. Within this subset $S$, the IFM oracles let each friend conclude $\square(b_A=1)$ and $\square(b_B=1)$ (each bomb was definitely live). The lab encodings then map these conclusions to $\square(C=0)$ and $\square(C=1)$, respectively, about the same upstream coin $C$. Under single-world facts and cross-agent consistency for necessity claims ($\square$), the subset $S$ entails a contradiction. Thus, CLF does not rely on numeric probabilities at all --- only on the consistency of possibilities ($\Diamond$) and necessities ($\square$) under the $\epsilon$-counterfactual inferences.

\subsection*{E.1. ABL calculation (three-box paradox)}
Preselect $|\psi_i\rangle = (|A\rangle + |B\rangle + |C\rangle)/\sqrt{3}$ 
and postselect $|\psi_f\rangle = (|A\rangle + |B\rangle - |C\rangle)/\sqrt{3}$. 
For a two–outcome test of $\Pi_A = |A\rangle\!\langle A|$, the 
Aharonov–Bergmann–Lebowitz rule gives the conditional probability 
as stated in Eq.~\eqref{eq:ABL}.

\noindent\textsc{Direct evaluation.}
\begin{align*}
\langle \psi_f | \Pi_A | \psi_i \rangle
&= \langle \psi_f | A \rangle\,\langle A | \psi_i \rangle
= \tfrac{1}{3}, \\[4pt]
\langle \psi_f | (I-\Pi_A) | \psi_i \rangle
&= \langle \psi_f | B \rangle\,\langle B | \psi_i \rangle \\
&\quad+ \langle \psi_f | C \rangle\,\langle C | \psi_i \rangle
= \tfrac{1}{3} - \tfrac{1}{3}
= 0~.
\end{align*}

Hence $P(\Pi_A=1\mid i,f)=1$. By symmetry $P(\Pi_B=1\mid i,f)=1$. 
Thus in this pre/post-selected ensemble one can be certain of $A$ 
and of $B$ simultaneously — the three-box paradox.

\subsection*{E.2. Mapping to $\epsilon$-IF probes}
Now replace the projective tests of $A$ and $B$ with $\epsilon$-counterfactual IFM oracles. That is, to test whether the particle is in box $A$, place an IFM bomb in arm $A$ (and similarly for $B$). When the oracle returns the decisive Dark outcome, the bomb’s reduced state change is bounded by $\epsilon$. Define $P_A$ (resp. $P_B$) as the observed probability of the Dark outcome in context $A$ (resp. $B$), conditioned on successful postselection. In the ideal $\epsilon \to 0$ limit with lossless devices, $P_A$ and $P_B$ coincide with the ABL values above (i.e., $P_A \approx P_B \approx 1$). For finite $\epsilon$ and small device inefficiencies, these probabilities remain $1 - O(\sqrt{\epsilon})$ (see Appendix B) and $1 - O(\text{loss})$, leaving the violation $P_A + P_B \approx 2$ intact for sufficiently small $\epsilon$.

\section{Full derivation of the $\epsilon$-IF three-box noncontextual bound}
See Section~\ref{sec:noncontextual-bound} for the main inequality statement.
This appendix expands the derivation in full detail.

\section{Composition of $\epsilon$ and Zeno scaling}
\subsection*{G.1. $\epsilon$ composition across rounds}
Consider a sequence of $m$ consecutive measurements (e.g. QEC syndrome extractions) acting on a target system, with per-round disturbance bounds $\epsilon_1, \ldots, \epsilon_m$ (in trace distance). By the triangle inequality and contractivity of trace distance under CPTP maps, the cumulative disturbance is bounded by $\epsilon_{\text{total}} \le \epsilon_1 + \cdots + \epsilon_m$. If each round occurs with probability near 1, then by concentration of measure the expected $\epsilon_{\text{total}}$ over a long cycle remains bounded by the same sum (up to $O(\sqrt{\epsilon_i})$ fluctuations from the conditionalization). This justifies a budgeting of disturbance \emph{additively} across rounds when scheduling low-back-action QEC checks.

\subsection*{G.2. Zeno-style success vs.\ absorbed dose}
Consider $N$ successive “weak looks,” each implemented by a small beam splitter or controlled-phase rotation (acting as the bomb), with mixing angle $\theta \ll 1$. For example, a sequence of $N$ weak absorptive measurements or $N$ small phase kicks. Tuning $\theta = \pi/(2N)$ (as in Kwiat’s Zeno interferometry scheme) yields dark-outcome success approaching $1 - O(1/N^2)$, while the cumulative absorption probability scales as $O(1/N)$ (in the ideal lossless limit). Consequently, to achieve a target disturbance $\epsilon$ one can choose $N$ large enough that the absorbed dose is $\le c_1\,\epsilon + c_2\,\text{(loss)}$, with device-dependent constants $c_1, c_2$ capturing interferometer loss and detector inefficiency. This underwrites the feasibility of $\epsilon$-certified low-dose IFM imaging.

\subsection*{G.3. Relating visibility to $\epsilon$}
In a simple model, the effect of a Dark outcome is to apply a dephasing channel 
on the bomb’s arm with coherence parameter $\lambda \in [0,1]$. 
The induced trace-distance change on a maximally sensitive bomb state is then 
rigorously bounded by
\begin{equation}
\epsilon \;\le\; 1 - |\lambda|~.
\end{equation}

Operationally, the measured interferometric visibility provides an estimate of 
$|\lambda|$: if $V_{0}$ is the visibility with the bomb removed and 
$V_{\text{dec}}$ the visibility in decisive-outcome runs, then 
$|\lambda| \approx V_{\text{dec}}/V_{0}$. Substituting this into the rigorous bound 
gives the conservative, approximate relation
\begin{equation}
\epsilon \;\lesssim\; 1 - \frac{V_{\text{dec}}}{V_{0}}~.
\label{eq:vis-eps}
\end{equation}

Thus, in practice one can report either (i) the rigorous state-level bound 
$\epsilon \le 1 - |\lambda|$, or (ii) the experimentally accessible proxy 
$\epsilon \lesssim 1 - V_{\text{dec}}/V_{0}$, with the latter understood as an 
empirical estimate subject to device imperfections and calibration.

\section{$\epsilon$-IF GHZ “all-versus-nothing” (AVN) paradox}
Three spatially separated labs ($A$, $B$, $C$) each host an IFM oracle controlled by a local “bomb bit” $b_j\in\{0,1\}$. We prepare the three-qubit GHZ state
\[
|GHZ\rangle = \frac{1}{\sqrt{2}}\bigl(|000\rangle + |111\rangle\bigr)
\]
on control registers (one per lab) that set phase shifts in each mediator’s interferometer. Each lab’s oracle writes a flag $W_j$ with value Dark ($D$) if and only if $b_j = 1$ (interaction-free, $\epsilon$-bounded write-out). We then choose four joint measurement settings on the labs — $XYY$, $YXY$, $YYX$, and $XXX$ — by inserting $\pi/2$ phase shifters in the appropriate interferometer arms before the final Hadamards. This has the effect of making each lab’s Dark/Bright output encode the eigenvalue ($-1$ or $+1$) of a Pauli operator ($X$ or $Y$) on the control qubit.

\textbf{Quantum predictions (deterministic).} 
Exactly as in the standard GHZ argument, the product of the three $\pm1$ outputs 
(interpreting $D = -1$, $B = +1$) equals $+1$ for the settings $XYY$, $YXY$, and $YYX$, 
and equals $-1$ for the setting $XXX$: 
\[
\langle GHZ | X\otimes X \otimes X | GHZ \rangle = -1.
\]

\textbf{AVN contradiction under single-world realism.} Assigning noncontextual $\pm1$ values to each lab’s two observables forces the product of the first three parity equations to be $+1$, which then requires the $XXX$ product to also be $+1$ — contradicting the quantum prediction of $-1$. Since every output is obtained via an interaction-free flag, the contradiction cannot be blamed on “measurement disturbance of the object.” Instead, the resolution lies in contextuality (or nonlocality) at a global level.

\textit{Remark:} This GHZ-IF construction provides an all-versus-nothing witness of contextuality/nonlocality, complementing the inequality-based three-box violation of Sec.~4.

\section{$\epsilon$-IF Peres–Mermin square (state-independent contextuality)}
We implement the nine observables of the Peres–Mermin magic square (two-qubit Pauli operators arranged in a $3\times 3$ grid whose row-wise products are $+1$ and column-wise products are $+1, +1, -1$) using IFM-controlled phase gadgets. Each $\pm1$ eigenvalue is encoded by a Dark/Bright flag with $\epsilon$-bounded disturbance on the tested qubits. Since the Peres–Mermin square is state-independent, any initial two-qubit state may be used.

For example, quantum mechanics predicts values such as
\[
\langle \psi | X\otimes X | \psi \rangle,\quad
\langle \psi | Y\otimes Y | \psi \rangle,\quad
\langle \psi | Z\otimes Z | \psi \rangle,
\]
obeying the algebraic constraints of the square.  

Noncontextual assignments that respect the local operator algebra would imply the product of all six measured parities (three rows and three columns) is $+1$, whereas the operator algebra demands a product of $-1$. Again, every measurement is carried out by an $\epsilon$-counterfactual flag, making it clear that disturbance to the tested qubits cannot resolve the contradiction.

\section{$\epsilon$-Leggett–Garg with noninvasive measurability replaced by IFM}
Leggett–Garg inequalities test macrorealism via time-separated measurements under the assumption of noninvasive measurability (NIM). We formulate an $\epsilon$-LG variant by replacing NIM with our $\epsilon$-counterfactual condition on decisive outcomes. Consider three time points $t_1 < t_2 < t_3$ and a dichotomic observable $Q(t) \in \{\pm 1\}$ (e.g., a two-level system’s $z$-spin). A macrorealist with $\epsilon$-NIM must satisfy 
\begin{equation}
\begin{split}
K_3 &\equiv C_{12} + C_{23} - C_{13} \\
    &\le 1 + c\,\epsilon~,
\end{split}
\label{eq:LG}
\end{equation}
where $C_{ij}$ denotes the two-time correlation
\[
C_{ij} = \langle Q(t_i)\, Q(t_j) \rangle
\]
measured from runs in which only those two times are probed by $\epsilon$-IFM flags.  

The coefficient $c$ depends on how $\epsilon$ composes under sequential measurements (cf. Appendix F). Quantum-coherent evolution interspersed with IFM probes can attain $K_3 \approx 1.5$ for small $\epsilon$, thus violating the bound. This supplies a macrorealism no-go statement in which measurement disturbance is explicitly bounded.

\textit{Derivation sketch.} We follow the standard derivation of the LG inequality but now include a total-variation term for switching measurement contexts, bounded by $K'\epsilon$ as in Appendix F. The algebra then yields the extra $+\,c\,\epsilon$ slack in the inequality.

\section{An epsilon-LF (Local Friendliness) inequality template}
We adapt the extended Wigner’s-friend no-go scenario of Bong et al.\ to include a finite $\epsilon$. Observers $A$ and $B$ (Wigners) each choose a setting $x\in\{0,1,2\}$ and $y\in\{0,1,2\}$, and their friends’ in-lab “measurements” are realized by interaction-free flag oracles with decisive outcomes. Define correlators
\[
E_{xy} = \sum_{a,b=\pm1} a\,b \, P(a,b|x,y) = \langle a\, b \rangle_{x,y}.
\]

Under the Local Friendliness assumptions of Absoluteness of Observed Events (AOE), Local Agency (LA), and No-Superdeterminism (NSD), the usual local-hidden-variable polytope yields linear bounds of the form
\begin{equation}
\begin{split}
S_{\text{LF}} \;\equiv\;&\; \sum_{x,y} c_{xy}\,E_{xy} \\
  \le\;&\; B_{\text{LF}} + K_1\,\epsilon \\
  &\quad + K_2\,\sqrt{\delta}~,
\end{split}
\label{eq:LF}
\end{equation}
for some integer coefficients $c_{xy}$ (see \cite{Bong2020,Wiseman2023} for concrete examples).  

Now suppose we have two such experiments that differ only by the insertion or removal of an $\epsilon$-counterfactual IFM module, whose Dark outcome occurs with probability $\ge 1-\delta$. By Lemma 1 (Appendix F) the prepared ontic state distributions differ by at most $K_1\epsilon + K_2\sqrt{\delta}$. It follows that the LF bound is relaxed to 
\begin{equation}
S_{\text{LF}} \;\le\; B_{\text{LF}} + K_1\,\epsilon + K_2\,\sqrt{\delta}~.
\end{equation}

Quantum correlations obtained by coherently entangled friends can exceed the original bound $B_{\text{LF}}$, and thus violate the relaxed bound for sufficiently small $\epsilon,\delta$. This provides an inequality-based complement to the CLF paradox, with an explicit continuity slack. *(For explicit choices of $c_{xy}$ and $B_{\text{LF}}$ see the cited works.)*

\section{Literature verification and novelty positioning}
We conclude by placing our contributions in context:
\begin{enumerate}
    \item \textbf{Local Friendliness (LF) canon.} A number of works have established and tested extended Wigner’s-friend paradoxes: Bong et al.~\cite{Bong2020} derive LF inequalities and report their violation; Proietti et al.~\cite{Proietti2019} perform an experimental test of local observer-independence; Wiseman et al.~\cite{Wiseman2023} reformulate the LF no-go with refined assumptions; Cavalcanti~\cite{Cavalcanti2021} connects LF violations to causal structures. None of these implement \emph{all} decisive inferences via interaction-free (IFM) flags with a provable $\epsilon$-bound on the probed object. Our CLF paradox is novel in precisely this sense.
    \item \textbf{Interaction-free elements in Wigner’s-friend contexts.} Waaijer and van Neerven analyze an extended Frauchiger–Renner scenario using interaction-free detection of records, but they do not achieve a Wigner’s-friend contradiction where every decisive inference is written by a unitary IFM oracle. Our CLF construction fills this gap in the literature \cite{Waaijer2021}.
    \item \textbf{Contextuality frameworks.} Our three-box inequality can be viewed as a specialization of known noncontextuality
inequalities (e.g. in the Spekkens framework \cite{Kunjwal2015,Kunjwal2018} or the
Cabello–Severini–Winter graph-theoretic approach \cite{Cabello2014}) to the case of
interaction-free probes, with an explicit continuity slack (epsilon-stability) to handle
the small measurement disturbance.
    \item \textbf{Continuity and gentleness tools.} To derive our $\epsilon$-stability results we leverage the gentle measurement lemma and modern “gentle sequential measurement” refinements \cite{Aaronson2017,Watts2024}, as well as continuity bounds like Fuchs–van de Graaf and Audenaert–Fannes\cite{Fuchs1999,Audenaert2007}.
    \item \textbf{IFM foundations.} Our $\epsilon$-counterfactual flag formalism abstracts the original interaction-free measurement schemes of Kwiat and collaborators~\cite{Kwiat1995,Kwiat1999} beyond their photonic context, and we cite these works as the operational ancestors of our approach.
\end{enumerate}

\section{Norm conversions and constants}
\begin{enumerate}[label=(\alph*)]
    \item \textbf{Diamond norm vs.\ state level.} If a bomb-conditional CPTP map $E$ obeys $\|E - \mathrm{id}\|_{\diamond} \le \epsilon_{\diamond}$, then for all input states (even entangled with an ancilla) the induced state change is $\le \epsilon_{\diamond}$ in trace distance (see Watrous~\cite{Watrous2018}, Thm.~3.55). Thus, one can always switch from a state-level $\epsilon$ certificate to a more conservative channel-level ($\diamond$-norm) one.
    \item \textbf{Fuchs–van de Graaf inequality.} For any two quantum states, $\,1 - F(\rho,\sigma) \;\le\; T(\rho,\sigma) \;\le\; \sqrt{\,1 - F(\rho,\sigma)^2}\,$, where $F$ is fidelity and $T$ is trace distance. This allows conversion of visibility or fidelity estimates into trace-distance disturbance bounds (up to second order terms) \cite{Fuchs1999}.
    \item \textbf{Gentle measurement lemma.} If an event is observed with probability $\ge 1-\delta$, the post-measurement state is $\le 2\sqrt{\delta}$ away (in trace distance) from the pre-measurement state. Further, any classical statistic extracted from such events changes by at most $O(\sqrt{\delta})$ under those conditions. Combining this with an $\epsilon$-counterfactual perturbation yields the $K_1\epsilon + K_2\sqrt{\delta}$ slack terms used in Appendices I, M, and N (see Aaronson~\cite{Aaronson2017} for a related result).
\end{enumerate}


\end{document}